\documentclass{article}
\usepackage[numbers, sort&compress]{natbib}
\usepackage{fullpage}

\usepackage{times}
\usepackage{soul}
\usepackage{url}
\usepackage[hidelinks]{hyperref}
\usepackage[utf8]{inputenc}
\usepackage[small]{caption}
\usepackage{graphicx}
\usepackage{amsmath}
\usepackage{amsthm}
\usepackage{booktabs}
\usepackage[switch]{lineno}

\usepackage{amssymb,mathtools,mathrsfs,bm,booktabs,autobreak}
\usepackage{braket}
\usepackage{nicematrix}
\NiceMatrixOptions{code-for-first-row = \scriptstyle\color{gray}, code-for-first-col = ~\scriptstyle\color{gray}}

\usepackage[ruled,vlined,linesnumbered,norelsize]{algorithm2e}
\mathtoolsset{showonlyrefs}
\usepackage{tikz}
\usetikzlibrary{backgrounds,calc}
\newtheorem{theorem}{Theorem} 
\newtheorem{lemma}{Lemma}
\newtheorem{proposition}{Proposition}

\newtheorem{definition}{Definition}
\newtheorem{example}{Example}

\newcommand{\conv}{\mathop{\rm conv}}
\newcommand{\supp}{\mathop{\rm supp}}
\newcommand{\argmax}{\mathop{\rm arg\,max}}

\newcommand{\cF}{\mathcal{F}}

\newcommand{\bA}{\bm{A}}
\newcommand{\cA}{\mathcal{A}}
\newcommand{\sA}{\mathscr{A}}

\newcommand{\sd}{\mathrm{sd}}
\newcommand{\ot}{\leftarrow}
\renewcommand{\epsilon}{\varepsilon}
\renewcommand{\mid}{\,:\,}

\usepackage{authblk}

\title{Random Assignment of Indivisible Goods under Constraints\thanks{A preliminary version appeared in the proceedings of IJCAI2023~\cite{KSY2023}.}}
\author[1]{Yasushi Kawase}
\author[2]{Hanna Sumita}
\author[2]{Yu Yokoi}
\affil[1]{University of Tokyo, Japan}
\affil[2]{Institute of Science Tokyo, Japan}
\date{}

\begin{document}
\maketitle




\begin{abstract}
We investigate the problem of random assignment of indivisible goods, in which each agent has an ordinal preference and a constraint. Our goal is to characterize the conditions under which there always exists a random assignment that simultaneously satisfies efficiency and envy-freeness. The probabilistic serial mechanism ensures the existence of such an assignment for the unconstrained setting. In this paper, we consider a more general setting in which each agent can consume a set of items only if the set satisfies her feasibility constraint. Such constraints must be taken into account in student course placements, employee shift assignments, and so on. We demonstrate that an efficient and envy-free assignment may not exist even for the simple case of partition matroid constraints, where the items are categorized, and each agent demands one item from each category. We then identify special cases in which an efficient and envy-free assignment always exists. For these cases, the probabilistic serial cannot be naturally extended; therefore, we provide mechanisms to find the desired assignment using various approaches.
\end{abstract}


\section{Introduction}
Assigning indivisible items to agents with preferences is one of the most fundamental problems in computer science and economics~\cite{AGT,rothe2015}.
Examples of such problems include university housing assignments, student course placements, employee shift assignments, and professional sports drafts.
In these kinds of problems, we are given a set of agents, a set of indivisible items, and preferences of the agents.
The goal of the problem is to find an assignment that satisfies efficiency and fairness. 
This study deals with the case where only ordinal information on preferences is available. Such an assumption is common in the literature because eliciting precise cardinal preferences would be difficult in practice (see Bogomolnaia and Moulin~\cite{BM2001} for more detailed justifications).

Randomization is frequently used to achieve both efficiency and fairness when assigning indivisible items.
Such a randomized assignment is referred to as \emph{lottery assignment}.
For example, lotteries are used in the draft process of many professional sports leagues.\footnote{Please refer to \url{https://en.wikipedia.org/wiki/Draft_(sports)}}
The standard way to define efficiency and fairness for a lottery assignment when only ordinal preferences are available is to use stochastic dominance (SD) relation. 
An agent prefers one lottery assignment over another in terms of the SD relation if she obtains at least as much utility on average from the former assignment as the latter for all possible cardinal utilities consistent with the revealed ordinal preference.

We consider \emph{sd-efficiency} as an efficiency concept, which states that no agent can be made better off without making at least one other agent worse off with respect to the SD relation. The sd-efficiency means efficiency in the ex ante sense and also leads to efficiency in the ex post sense~\cite{BM2001}.
Additionally, as a concept of fairness, we consider \emph{sd-envy-freeness}, which states that every agent prefers her (ex ante) assignment to that of every other agent with respect to the SD relation. 
Note that the sd-envy-freeness guarantees fairness in the ex ante sense but not in the ex post sense. The ex post unfairness is inevitable in the assignment of indivisible items.
We also examine some other efficiency and fairness criteria.
In a random assignment problem in which each agent receives one object, Bogomolnaia and Moulin~\cite{BM2001} proposed the \emph{probabilistic serial (PS)} mechanism. 
In the mechanism, agents ``eat'' their preferred goods at an equal rate until all goods are consumed.
This outputs a lottery assignment that is both sd-efficient and sd-envy-free.
Kojima~\cite{kojima2009} generalized this result to the case where each agent can receive more than one item and the agents' preferences are additively separable over the items.

Note that these studies focused on the unconstrained case.
In reality, however, assignment problems frequently involve constraints. Motivated by real-world applications such as refugee resettlement~\cite{delacretaz2016refugee}, college admissions with budget constraints~\cite{abizada2016stability}, 
and day-care allocation~\cite{okumura2019school}, assignment (or matching) problems under constraints have recently been an active research subject. 
As for random assignment under constraints,
Aziz and Brandl~\cite{AB2022} proposed a generalized PS mechanism, called \emph{vigilant eating rule} (VER), for a constrained case.
This mechanism produces a random assignment that satisfies sd-efficiency and equal treatment of equals, which is a weaker fairness notion than our sd-envy-freeness. However, VER may not produce an sd-envy-free lottery assignment.

In this study, we seek to attain sd-efficiency and sd-envy-freeness in a general setting where each agent can consume a set of items only if it satisfies her feasibility constraints. 
We assume that constraints satisfy the \emph{hereditary property}, meaning that any subset of a feasible set is also feasible.
A typical example of the hereditary property is a knapsack constraint, which represents the capacity of a limited resource, such as budget, time, or space.
We have a particular interest in matroid constraints, which is a subclass of hereditary constraints and is a generalization of constraints that arise in the multi-type resource allocation problem.
The class of matroids is expressive enough to represent various constraints that naturally arise in many real-life assignment problems.
For example, in the context of weekly employee shift assignments, if an employee can work at most one time slot per day, then her feasibility constraint is represented by a partition matroid.
Even if she additionally declares that she can work at most three days a week, then her feasibility constraint is still a matroid (for the formal definition, see Model section). 
Furthermore, it is known that matroid structure provides fruitful results in many other related assignment or matching problems~\cite{babaioff2020fair,benabbou2020,BarmanVermaAAMAS2021,goko2024fair,IK2022}.

This study aimed to identify the settings in which sd-efficiency and sd-envy-freeness are compatible.
We demonstrate that an sd-efficient and sd-envy-free lottery assignment may not exist even for the simple case of partition matroid constraints. We then identify special cases in which an sd-efficient and sd-envy-free lottery assignment always exists. Moreover, for such cases, we provide mechanisms to find the desired lottery assignment.
This study does not focus on the strategic issue because no mechanism simultaneously satisfies sd-efficiency, sd-envy-freeness, and sd-weak-strategy-proofness (see \ref{sec:SP} for details).

\subsection{Our contributions}
We investigate the existence of sd-efficient and sd-envy-free assignments in 16 settings according to the following:
(i) the number of agents is $2$ or arbitrary $n$,
(ii) the constraints are matroids or general hereditary constraints,
(iii) the constraints of the agents are identical or heterogeneous, and 
(iv) the ordinal preferences of the agents are identical or heterogeneous.

\begin{figure*}
\centering
\definecolor{myred}{rgb}{1,0.5,0.5}
\definecolor{mygreen}{rgb}{0.4,.95,0.4}
\scalebox{.9}{
\begin{tikzpicture}[scale=.95]
  \tikzset{
    block/.style={rectangle, draw=black, inner sep=4pt,align=center, text centered, text width=1.25cm, text height=.2cm,font=\small},
    label/.style={font=\small, fill=white, fill opacity=.7, text opacity=1}
  };
  \foreach \a/\x in {0/i,1/h} {
    \foreach \b/\y in {0/i,1/h} {
      \foreach \c/\z in {0/m,1/g} {
        \foreach \d/\w in {0/$2$,1/$n$} {
          \node[block] (\a\b\c\d) at ($\a*(-2,1.3)+\b*(0,1.3)+\c*(2,1.3)+\d*(7,1.3)$) {\w,\z,\x,\y};
        }
      }
    }
  }
  \begin{scope}[on background layer]
  \foreach \a in {0,1} {
    \foreach \b in {0,1} {
      \foreach \c in {0,1} {
        \draw[gray] (0\a\b\c) -- (1\a\b\c);
        \draw[gray] (\a0\b\c) -- (\a1\b\c);
        \draw[gray] (\a\b0\c) -- (\a\b1\c);
        \draw[gray] (\a\b\c0) -- (\a\b\c1);
      }
    }
  }
  \node[label, xshift=-.7cm, yshift=-0.45cm] at (1100) {Theorem~\ref{thm:2-mat-het-het}};
  \node[label, xshift=-.7cm, yshift=-0.45cm] at (1001) {Theorem~\ref{thm:n-mat-het-id}};
  \node[label, xshift=.7cm, yshift=-0.45cm] at (0011) {Theorem~\ref{thm:n-gen-id-id}};
  \node[label, yshift=0.45cm] at (1010) {Theorem~\ref{thm:2-gen-het-id}};
  \node[label, yshift=0.45cm] at (0101) {Theorem~\ref{thm:3-mat-id-het}};

  \foreach \a/\b in {1100/1000,1100/0100,1001/0001,1001/1000,0011/0010,0011/0001,0001/0000,0010/0000,0100/0000,1000/0000} \draw[ultra thick, mygreen, opacity=.5] (\a) -- (\b);
  \foreach \a/\b in {1010/1110,1010/1011,0101/1101,0101/0111,1110/1111,1101/1111,1011/1111,0111/1111} \draw[ultra thick, myred, opacity=.5] (\a) -- (\b);
  
  \foreach \a in {1100,1001,0011} \node[block,fill=mygreen,draw=mygreen,ultra thick] at (\a) {};
  \foreach \a in {0000,0001,0010,0100,1000} \node[block,fill=mygreen!10,draw=mygreen,ultra thick] at (\a) {};

  \foreach \a in {1010,0101} \node[block,fill=myred,draw=myred,ultra thick] at (\a) {};  
  \foreach \a in {1111,1110,1101,1011,0111} \node[block,fill=myred!10,draw=myred,ultra thick] at (\a) {};

  \node[block,fill=white,draw=white,ultra thick] at (0110) {};
  \end{scope}
\end{tikzpicture}}
\caption{Summary of our results on the existence of an sd-efficient and sd-envy-free assignment. Each of the $16$ cases is identified by four characters, such as ``$2$,m,i,i.'' The first, second, third, and fourth characters, respectively, indicate whether there are $2$ or an arbitrary $n$ number of agents, whether the constraints are matroids or general, whether the constraints are identical or heterogeneous, and whether the preferences are identical or heterogeneous. For each case, the box is painted green if such a lottery assignment always exists and red otherwise.}\label{fig:our_results}
\end{figure*}
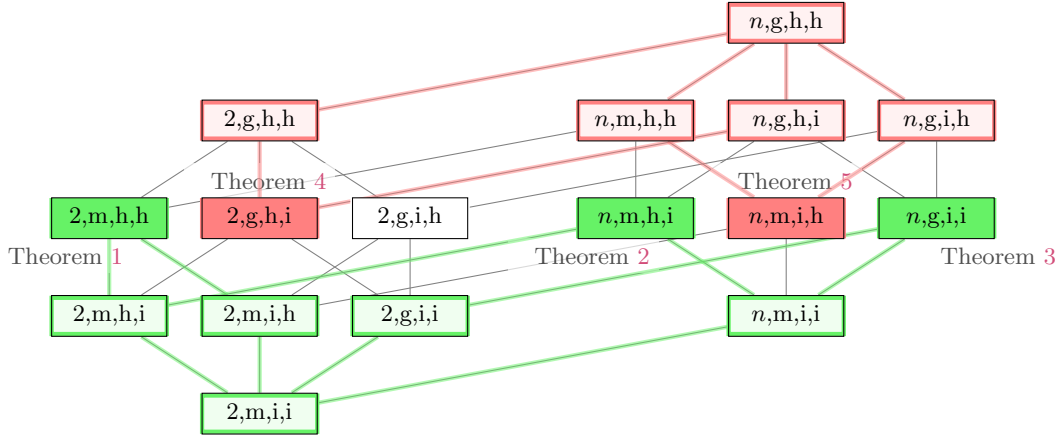

We demonstrate the impossibility of an sd-efficient and sd-envy-free assignment even when there are either
\begin{itemize}
    \item 
    $2$ agents with identical preferences
    (Theorem~\ref{thm:2-gen-het-id}) or 
    \item 
    $3$ agents with identical partition matroid constraints
    (Theorem~\ref{thm:3-mat-id-het}).
\end{itemize}
As tractability results, we demonstrate that an sd-efficient and sd-envy-free assignment always exists when there are
\begin{itemize}
    \item 
    $2$ agents with matroid constraints
    (Theorem~\ref{thm:2-mat-het-het}),
    \item 
    agents with identical preferences and heterogeneous matroid constraints
    (Theorem~\ref{thm:n-mat-het-id}), or 
    \item 
    identical agents 
    (Theorem~\ref{thm:n-gen-id-id}).
\end{itemize}
Moreover, we provide polynomial-time algorithms that find desired assignments in the settings of Theorems~\ref{thm:2-mat-het-het} and \ref{thm:n-mat-het-id}.
By considering the inclusion relation, we obtain the results shown in Figure~\ref{fig:our_results}.
The existence of an sd-efficient and sd-envy-free assignment is open when there are $2$ agents with identical constraints.

In terms of computational complexity, determining the existence of an sd-efficient and sd-envy-free lottery assignment is NP-complete even when there are two agents with identical preferences and constraints (Theorem~\ref{thm:NPhard}).
We provide an exponential-time algorithm to find an sd-efficient and sd-envy-free lottery assignment if it exists for the general setting.

We investigate \emph{possible-envy-freeness}, \emph{anonymity}, \emph{necessarily Pareto-efficiency}, and \emph{sd-proportionality} as other efficiency and fairness notions (see Section~\ref{subsec:desirable_properties} and \ref{app:others}).

\subsection{Related work}
Random assignment problems under partition matroid constraints are studied under the name of multi-type resource allocation problem~\cite{MT2015,MX2016,WSG+2020,GSW+2021}.
Different from our setting, these studies deal with a case where the preferences of agents are defined over bundles.

The above-mentioned works and the present work deal with constraints on the agent side. There are also studies on random assignment with constraints on the item side ~\cite{FSZ2018,BCKM2013}.
Fujishige et al.~\cite{FSZ2018} provided an extension of the PS mechanism that outputs an sd-efficient and sd-envy-free assignment if the set of feasible integral vectors of items forms an integral polymatroid.
Their proof heavily depends on the (generalized) Birkhoff--von Neumann theorem.
Note that the theorem also holds for our problem if the constraints are matroids.
This leads us to expect that the PS mechanism produces an sd-efficient and sd-envy-free assignment when the constraints are matroids, but this is not the case, as we will show in Theorem~\ref{thm:3-mat-id-het}.

For the cardinal case, Cole and Tao~\cite{CT2021} proved that there always exists a Pareto-efficient and envy-free lottery assignment. 
Their framework is so general that any partition-based utility functions (including additive utility functions with any constraints) can be handled. 
Their proof is based on fixed-point arguments and does not imply polynomial-time algorithms. 
Kawase and Sumita~\cite{KawaseSumita2020, KS2024} studied the computational complexity of finding a max-min fair lottery assignment under envy-free constraint in a cardinal setting.

\section{Model}\label{sec:model}
For a nonnegative integer $k$, we write $[k]$ to denote $\{1,2,\dots,k\}$.
An instance of our problem is a tuple $(N,E,({\succ_i}, \cF_i)_{i\in N})$, where
$N=[n]$ represents the set of agents and $E=\{e_1,e_2,\dots,e_m\}$ represents the set of indivisible items.
Each agent $i\in N$ has a strict preference $\succ_i$ over $E$ and can consume a set of items in $\cF_i\subseteq 2^E$, which is the feasible set family of agent $i$.
We assume that $\cF_i$ is given by a membership oracle for each $i\in N$.
The preferences over sets of items are additively separable across items, meaning that each agent $i$ has a cardinal utility function $u_i\colon E\to\mathbb{R}_{++}$, and 
her utility for a bundle $E'\in\cF_i$ is $\sum_{e\in E'}u_i(e)$.
Here, $\mathbb{R}_{++}$ is the set of positive real numbers.
We assume that the preference of each agent $i$ has no ties 
and that the central authority knows only the ordinal preferences $\succ_i$ that are consistent with $u_i$. In other words, $\succ_i$ is a strict order on $E$ such that $e\succ_i e'$ if and only if $u_i(e)>u_i(e')$.

For each agent $i\in N$, we assume that the pair $(E, \cF_i)$ forms an \emph{independence system}: the feasible set family $\cF_i\subseteq 2^E$ is nonempty and satisfies the \emph{hereditary property}, that is, $X\subseteq Y\in\cF_i$ implies $X\in\cF_i$. 
We denote by $\conv(\cF_i)$ the convex hull of the characteristic vectors of the members of $\cF_i$, where the characteristic vector of $X\in\cF_i$ is 
a vector in $\{0,1\}^E$ whose component corresponding to $e\in E$ is $1$ if and only if $e\in X$,
and the convex hull of $S\subseteq \mathbb{R}^E$ is the smallest convex set containing $S$.
We note that $\conv(\cF_i) \subseteq [0,1]^E$. 
We will also consider a special case where each $(E,\cF_i)$ is a \emph{matroid}, which is an independence system satisfying a property called the \emph{augmentation axiom}: if $X,Y\in\cF_i$ and $|X|<|Y|$ then there exists $e\in Y\setminus X$ such that $X\cup\{e\}\in\cF_i$.
A simple example of a matroid is a {\em partition matroid}, which represents a constraint in which items are categorized, and the number of items we can take from each category is constrained. More precisely, a partition matroid $(E, \cF)$ is determined by a partition $E_1, E_2,\dots, E_k$ of $E$ and capacities $q_1, q_2, \dots, q_k\in \mathbb{Z}_+$, and $\cF$ is of the form $\{X\subseteq E\mid |X\cap E_i|\leq q_i~(\forall i \in [k])\}$.

A deterministic {\em assignment} is a list $\bA=(A_1,\dots,A_n)$ of subsets of $E$ such that 
(i) $A_i\in \cF_i$ for all $i\in N$ and
(ii) $A_i\cap A_j=\emptyset$ for all distinct $i,j\in N$.
Let $\cA$ be the set of all deterministic assignments.
A \emph{lottery assignment} is a probability distribution over $\cA$.
We denote the set of all lottery assignments by $\Delta(\cA)$.

A \emph{fractional assignment} is a matrix $\pi=(\pi_{ie})_{i\in N, e\in E}\in\mathbb{R}^{N\times E}$ such that, for every item $e\in E$,  $\sum_{i\in N} \pi_{ie}\leq 1$. We interpret $\pi_{ie}$ as the probability that agent $i\in N$ receives item $e\in E$. 
For each $i\in N$, we denote the row in $\pi$ corresponding to agent $i$ by $\pi_i$, that is, $\pi_i = (\pi_{ie})_{e\in E} \in [0,1]^E$.
A lottery assignment $p\in \Delta(\cA)$ induces a fractional assignment $\pi\in\mathbb{R}^{N\times E}$ such that $\pi_{ie}=\Pr_{\bA\sim p}[e\in A_i]=\sum_{\bA\in\cA:\, e\in A_i}p_{\bA}$ for all $i\in N$ and $e\in E$. We will write $\pi^p$ for the fractional assignment induced from $p$. A fractional assignment is called {\em feasible} if it is induced from some lottery assignment.

Let $\conv(\cA)\subseteq \mathbb{R}^{N\times E}$ be the convex hull of the characteristic vectors of the members of $\cA$.
By definition, a fractional assignment belongs to $\conv(\cA)$ if and only if it is feasible, i.e., induced from some lottery assignment $p\in \Delta(\cA)$.
For any feasible fractional assignment $\pi\in \conv(\cA)$, a lottery assignment inducing $\pi$ is not unique in general.
According to Carath\'eodory's theorem~(see, e.g., Schrijver~\cite[p.94]{Schrijver1998}), there exists such a lottery assignment with a support size of not more than $|N|\cdot |E|+1$.

\subsection{Desirable properties}\label{subsec:desirable_properties}
For a preference $\succ_i$, let $U(\succ_i,e)\coloneqq \{e'\in E\mid e'\succeq_i e\}$ be the set of items that are not worse than $e$ with respect to $\succ_i$.
We say that $x\in\mathbb{R}_+^E$ \emph{weakly stochastically dominates} $y\in\mathbb{R}_+^E$, denoted by $x\succeq_i^\sd y$, if $\sum_{e'\in U(\succ_i,e)} x_{e'}\ge \sum_{e'\in U(\succ_i,e)} y_{e'}$ for all $e\in E$.
If $x\succeq_i^\sd y$ and $x\ne y$, we say that $x$ \emph{stochastically dominates} $y$ and denote $x\succ_i^\sd y$.
Note that $x$ stochastically dominates $y$ if and only if the expected utility of $x$ is greater than that of $y$ for all possible cardinal utilities consistent with $\succ_i$.

Pareto-efficiency is a standard efficiency concept where no agents can be made better off without making at least one other agent worse off.
A natural notion of efficiency for lottery assignments is defined as Pareto-efficiency with respect to the SD relation.
\begin{definition}[sd-efficiency]\label{def:sd-efficient}
A lottery assignment $p\in \Delta(\cA)$ is called \emph{sd-efficient} (also called \emph{ordinally efficient} or \emph{necessarily Pareto-efficient}) if 
there is no lottery assignment $q\in\Delta(\cA)$ that satisfies $\pi_i^q \succeq_i^\sd \pi_i^p$ for all $i\in N$ and $\pi_j^q \succ_j^\sd \pi_j^p$ for some $j\in N$.
\end{definition}
Note that, for any lottery assignment $p\in \Delta(\cA)$, we have $\sum_{e'\in U(\succ_i,e)} \pi^p_{ie'}=\sum_{e'\in U(\succ_i,e)}\sum_{\bA\in\cA:\,e'\in A_i} p_{\bA}=\sum_{\bA\in\cA}p_{\bA}|A_i\cap U(\succ_i,e)|$, and hence the condition $\pi_i^q\succeq_i^\sd \pi_i^p$ in Definition~\ref{def:sd-efficient} is equivalent to the condition 
\begin{align}
\sum_{\mathclap{\bA\in\cA}}q_{\bA}|A_i{\cap} U(\succ_i,e)|\ge
\sum_{\mathclap{\bA\in\cA}}p_{\bA}|A_i{\cap} U(\succ_i,e)|\ \ (\forall e\in E).
\end{align}
In addition, a lottery assignment is sd-efficient if and only if it maximizes utilitarian social welfare for some possible cardinal utilities consistent with $(\succ_i)_{i\in N}$.

A weaker notion of efficiency can be defined as an ex post sense.
A lottery assignment $p\in \Delta(\cA)$ is called \emph{ex post efficient} if, for any $\bA\in \cA$ with $p_{\bA}>0$, a lottery assignment that takes $\bA$ with probability $1$ is sd-efficient.
By definition, sd-efficiency implies ex post efficiency. On the other hand, ex post efficiency does not imply sd-efficiency.
\begin{example}[\mbox{ex post efficiency does not imply sd-efficiency}]\label{ex:simple}
Consider an instance of our problem $(N,E,(\succ_i, \cF_i)_{i\in N})$ where
$N=\{1,2\}$, $E=\{a,b,c,d\}$, $\cF_i=\{X\subseteq E \mid |X\cap\{c,d\}|\le 1 \text{ and }|X|\le 2\}$, and $a \succ_i b \succ_i c \succ_i d$ for both $i=1,2$.
Note that $(E,\cF_1)~(=(E,\cF_2))$ is a matroid.

Let $p$ be the lottery assignment that takes each of deterministic assignments $(A_1, A_2)=(\{a,c\}, \{b,d\})$ and $(\{b,d\}, \{a,c\})$ with probability $0.5$. 
Also, let $q$ be the lottery assignment that takes each of $(A_1, A_2)=(\{a,b\}, \{c\})$ and $(\{c\}, \{a,b\})$ with probability $0.5$. 
It is not difficult to check that the lottery assignments $p$ and $q$ are ex post efficient.
Note that $p$ and $q$ respectively induce the following fractional assignments:
\begin{align}
\pi^{p}=
{\begin{pNiceMatrix}[first-row,first-col]
& a & b & c & d \\
1  & 1/2 & 1/2 & 1/2 & 1/2\\[2mm]
2  & 1/2 & 1/2 & 1/2 & 1/2
\end{pNiceMatrix}}
~~\text{and}~~
\pi^{q}=
{\begin{pNiceMatrix}[first-row,first-col]
& a & b & c & d \\
1  & 1/2 & 1/2 & 1/2 & 0\\[2mm]
2  & 1/2 & 1/2 & 1/2 & 0
\end{pNiceMatrix}}.
\end{align}
Thus, $q$ is not sd-efficient because 
$\pi_i^p\succeq^\sd \pi_i^q$ for all $i\in N$.
On the other hand, $p$ is sd-efficient.
\end{example}

\emph{Necessarily Pareto-efficient} is a stronger notion of efficiency, which means that a lottery assignment is Pareto-efficient under every possible cardinal utility consistent with the given ordinal preferences.
This notion is outside the scope of this paper because no assignment may satisfy it, even in the single-agent case (Proposition~\ref{prop:no-necessaryPO}). 

As a notion of fairness, we consider envy-freeness. 
For the unconstrained setting, a standard definition of sd-envy-freeness requires a fractional assignment to satisfy $\pi^p_i\succeq_i^\sd \pi^p_j$ for any agents $i,j\in N$. This condition is equivalent to the expected utility of the fractional assignment of agent $i$ being no worse than that of any other agent $j$ with respect to any cardinal utility consistent to $\succ_i$~\cite{AGMW2015}.
In our setting, however, this equivalence does not hold due to the existence of constraints. Indeed, the bundle assigned to agent $j$ is not feasible for agent $i$ in general. Therefore, we have to take constraints into account when considering each agent's envy toward other agents. 
For a utility function $u_i$ consistent to $\succ_i$, let $\tilde{u}_i(X)$ be $i$'s evaluation of a bundle $X\subseteq E$ (that may be infeasible for $i$ to consume). That is, $\tilde{u}_i(X)= \max\set{\sum_{e\in Y}u_i(e)|Y\subseteq X,\ Y\in \cF_i}$.
Then, a natural generalization of sd-envy-freeness is to impose a lottery assignment $p\in\Delta(\cA)$ to satisfy
\begin{align}
\begin{split}
\mathbb{E}_{\bA\sim p}&[\tilde{u}_i(A_i)] \ge \mathbb{E}_{\bA\sim p}[\tilde{u}_i(A_j)] \quad
(\forall i, j\in N, \forall u_i\in \mathbb{R}_{++}^{E} \text{ consistent to } \succ_i).
\end{split}\label{eq:sdEF0}
\end{align}
It turns out that the condition \eqref{eq:sdEF0} is equivalent to the condition \eqref{eq:sdEF} below.
Since \eqref{eq:sdEF} does not use utility functions,
we adopt \eqref{eq:sdEF} as the definition of sd-envy-freeness.
We show the equivalence to \eqref{eq:sdEF0} in Proposition~\ref{prop:sdEF}.
The envy-freeness with respect to the SD relation is defined as follows.
\begin{definition}[sd-envy-freeness]\label{def:sd-envy-free}
A lottery assignment $p\in\Delta(\cA)$ is called \emph{sd-envy-free} (also called \emph{necessary envy-free} or \emph{not envy-possible}) if
\begin{align}
\begin{split}
\sum_{\bA\in\cA}p_{\bA}|A_i\cap U(\succ_i, e)|\ge \sum_{\bA\in\cA}p_{\bA}\max_{\substack{\scriptscriptstyle Y\subseteq A_j: \\\scriptscriptstyle Y\in\cF_i}}|Y\cap U(\succ_i, e)| \quad
(\forall i,j\in N,\ \forall e\in E).
\end{split}
\label{eq:sdEF}
\end{align}
\end{definition}
Note that, if the constraints are identical (i.e., $\cF_1=\dots=\cF_n$), the condition \eqref{eq:sdEF} coincides with $\pi^p_i\succeq_i^\sd \pi^p_j$ (recall that $\sum_{e'\in U(\succ_i,e)} \pi^p_{ie'}=\sum_{\bA\in\cA}p_{\bA}|A_i\cap U(\succ_i,e)|$).
Hence, Definition~\ref{def:sd-envy-free} indeed generalizes the standard definition of sd-envy-freeness.
\begin{proposition}\label{prop:sdEF}
A lottery assignment $p\in\Delta(\cA)$ is sd-envy-free if and only if it satisfies \eqref{eq:sdEF0}.
\end{proposition}
The proof of this fact can be found in \ref{app:sdEF}.

In contrast to sd-efficiency, which is defined only by the induced fractional assignment $\pi^p$, the definition of sd-envy-freeness requires the information of a lottery assignment $p$ itself. 
That is, sd-envy-freeness is a property of lottery assignments and cannot be that of fractional assignments in our constrained setting.
The following example gives two lottery assignments that induce the same fractional assignment, but only one of them is sd-envy-free.
\begin{example}
Consider an instance $(N,E,(\succ_i, \cF_i)_{i\in N})$ where
$N=\{1,2\}$, $E=\{a,b\}$, $\cF_1=\big\{\emptyset,\{a\},\{b\}\big\}$, $\cF_2=\big\{\emptyset,\{a\},\{b\},\{a,b\}\big\}$, and $a \succ_i b$ for $i=1,2$.
Then, the lottery assignment $p$ that takes each of deterministic assignments $(A_1,A_2)=(\{a\},\emptyset)$ and $(\emptyset,\{a,b\})$ with probability $0.5$ is sd-envy-free.
In contrast, the lottery assignment $q$ that takes each of deterministic assignments $(A_1,A_2)=(\{a\},\{b\})$ and $(\emptyset,\{a\})$ with probability $0.5$ is not sd-envy-free because 
$\sum_{\bA\in\cA}q_{\bA}|A_1\cap U(\succ_i, b)|=0.5$ is smaller than
$\sum_{\bA\in\cA}q_{\bA}\max_{Y\subseteq A_2:\,Y\in\cF_1}|Y\cap U(\succ_i, b)|=1$.
However, the two lottery assignments lead to the same fractional assignment, i.e., $\pi^p=\pi^q$.
\end{example}

We call a lottery assignment $p$ \emph{possible-envy-free} if, for each agent $i$, there exists a cardinal utility $u_i\in\mathbb{R}_{++}^E$ consistent to $\succ_i$ such that $i$ does not envy any other agent in terms of expectation (i.e., $\mathbb{E}_{\bA\sim p}[\tilde{u}_i(A_i)] \ge \mathbb{E}_{\bA\sim p}[\tilde{u}_i(A_j)]$ for all $j\in N$).
Hence, a lottery assignment $p$ is sd-efficient and possible-envy-free if there exist consistent cardinal utilities that make $p$ Pareto-efficient and envy-free in the cardinal sense. 
Combined with a fact known for a cardinal setting, this implies the existence of an sd-efficient and possible-envy-free lottery assignment in our setting. That is, to find such a lottery assignment, it is sufficient to fix a consistent utility function for each agent (say, $u_i(e)=|\{e'\in E\mid e\succeq e'\}|$ for each $i\in N$ and $e\in E$) and take a lottery assignment that satisfies Pareto-efficiency and envy-freeness with respect to these utility functions, where the existence of such an assignment is guaranteed~\cite{CT2021}.


For the unconstrained setting, the PS mechanism is known to satisfy both sd-efficiency and sd-envy-freeness.
Therefore, to achieve these two desirable properties in our constrained setting, a natural approach is to consider the generalized version of the PS mechanism in which each agent consumes the best remaining item while preserving feasibility. 
The vigilant eating rule (VER) mechanism in \cite{AB2022} includes this generalization. 
However, the generalized PS mechanism does not guarantee sd-envy-freeness, as shown in the following example.
\begin{example}[generalized PS is not sd-envy-free]\label{ex:PSisBad}
Consider an instance $(N,E,(\succ_i, \cF_i)_{i\in N})$ where
$N=\{1,2\}$, $E=\{e_1,e_2,\dots,e_5\}$, $e_1\succ_i e_2\succ_i\dots\succ_i e_5~(i=1,2)$, 
\begin{align}
    \cF_1&=\{X\subseteq E\mid |X\cap\{e_1,e_2,e_3,e_5\}|\le 2\},\ \text{and}\\
    \cF_2&=\{X\subseteq E\mid |X\cap\{e_1,e_2,e_3\}|\le 1\}.
\end{align}
Note that $(E,\cF_1)$ and $(E,\cF_2)$ are matroids.
For this instance, the generalized PS mechanism (the VER mechanism) outputs 
\[
\pi = 
{\begin{pNiceMatrix}[first-col,first-row]
   & e_1 & e_2 & e_3 & e_4 & e_5 \\
1  & 1/2 & 1/2 & 1 & 0 & 0\\
2  & 1/2 & 1/2 & 0 & 1 & 1\\
\end{pNiceMatrix}}.
\]
However, this is not sd-envy-free because the total amount of items assigned to agent $2$ is larger than that of agent $1$, which causes agent $1$ to envy agent $2$.
\end{example}

\section{Related Properties of Matroids}
We introduce several matroid properties, which will be used in our analyses.

Let $(E,\cF)$ be a matroid and $\conv(\cF)\subseteq \mathbb{R}_+^E$ be the convex hull of the characteristic vectors of the members of $\cF$.
For two vectors $x,y\in\mathbb{R}^E$, we write $y\le x$ if $y_e\le x_e$ for all $e\in E$.
For any vector $x\in \mathbb{R}_+^E$, define a polytope $\conv(\cF)^x \coloneqq\set{y\in \conv(\cF) \mid y\leq x}\subseteq \mathbb{R}_+^E$. 
Recall that, for any total order $\succ$ on $E$, we denote $s\succ^\sd t$ if $s\in \mathbb{R}_+^E$ stochastically dominates $t\in \mathbb{R}_+^E$ with respect to $\succ$.

We call a vector $s\in\conv(\cF)^x$ \emph{lexicographically maximum} with respect to $\succ$ in $\conv(\cF)^x$ if the value of the highest rank component is as large as possible; subject to this, the value of the next highest rank component is as large as possible, and so on.
The following lemma shows that 
the lexicographically maximum vector stochastically dominates all other vectors in the polytope $\conv(\cF)^x$.
Note that this is a special property of a matroid and is not satisfied in general if $(E, \cF)$ is an arbitrary independence system.%
\footnote{Consider the convex hull of the non-matroid family 
$\cF=\big\{\emptyset, \{e_1\}, \{e_2\}, \{e_3\}, \{e_2,e_3\}\big\}$ and the total order $\succ$ such that $e_1\succ e_2\succ e_3$, and let $x=(x_{e_1},x_{e_2},x_{e_3})=(1,1,1)$. The lexicographically maximal solution is $y^*=(1,0,0)\in \conv(\cF)^x$, but $y^*$ does not stochastically dominate $(0,1,1)\in \conv(\cF)^x$.}
\begin{lemma}\label{lem:lexopt=sdopt}
For any matroid $(E,\cF)$, any vector $x\in \mathbb{R}_+^E$, and any total order $\succ$ on $E$, let $y^*$ be the lexicographically maximum vector in $\conv(\cF)^x$ with respect to $\succ$. Then, $y^*\succeq^\sd y$ holds for every vector $y\in \conv(\cF)^x$.
\end{lemma}
We need some preparation for the proof of Lemma~\ref{lem:lexopt=sdopt}. We first give some basic facts on matroids.

The {\em rank function} $r:2^E\to \mathbb{R}$ of matroid $(E,\cF)$ is defined by $r(X)=\max\{|Y|\mid Y\subseteq X,~Y\in \cF\}$. 
This function satisfies the following properties.
\begin{itemize}
\item $r(X)\in \mathbb{Z}$ and $0\leq r(X)\leq |X|$ for any $X\subseteq E$, 
\item $X\subseteq Y\subseteq E$ implies $r(X)\leq r(Y)$ (monotonicity), and
\item $r(X)+r(Y)\geq r(X\cup Y)+r(X\cap Y)$ for any $X, Y\subseteq E$ (submodularity).
\end{itemize}
Note that, for each $X\subseteq E$, the rank $r(X)$ can be computed in polynomial time by the greedy algorithm with the membership oracle~\cite{Schrijver2003}.
A matroid rank function is a special case of a {\em polymatroid rank function}, which requires monotonicity, submodularity, and $r(\emptyset)=0$ (that is, only the second and third conditions above are required).
Using the rank function of matroid $(E, \cF)$, the polytope $\conv(\cF)$ is expressed as follows:
\[\conv(\cF)=\set{y\in \mathbb{R}_+^E| y(X)\leq r(X) \ (\forall X\subseteq E)},\]
where $y(X)=\sum_{e\in X}y_e$.

For any vector $x\in \mathbb{R}_+^E$,
recall that a polytope $\conv(\cF)^x$ is defined as 
\[\conv(\cF)^x\coloneqq\set{y\in \conv(\cF) \mid y\leq x}\subseteq \mathbb{R}_+^E.\]
Let us define a function $r^x\colon 2^E\to \mathbb{Z}$ by 
$r^x(X)=\min\set{r(Y)+x(X\setminus Y) | Y\subseteq X}$.
The following statement is a specialization of the well-known fact on polymatroid.
\begin{proposition}[{Fujishige~\cite[p.47]{Fujibook}}]\label{prop:polymatroid-reduction}
For any matroid $(E, \cF)$ whose rank function is $r$ and any vector $x\in \mathbb{R}_+^E$, the function $r^x$ is a polymatroid rank function and 
\[\conv(\cF)^x=\set{y\in \mathbb{R}_+^E| y(X)\leq r^x(X)~(\forall X\subseteq E)}.\]
\end{proposition}

Now, we are ready to show Lemma~\ref{lem:lexopt=sdopt}, which states that the lexicographically maximum vector in $\conv(\cF)^x$ stochastically dominates all other vectors in $\conv(\cF)^x$.
\begin{proof}[Proof of Lemma~\ref{lem:lexopt=sdopt}]
Assume that the elements of $E=\{e_1, e_2, \dots, e_m\}$ satisfy $e_1\succ e_2\succ \cdots \succ e_m$ without losing generality.
For each $k \in [m]$, define $E_k\coloneqq\{e_1,e_2,\dots,e_k\}$, that is, $E_k=\set{e\in E|e\succeq e_k}$ and $E_0=\emptyset$.
Define a vector $y^*\in \mathbb{R}_+^E$ by $y^*(e_i)=r^x(E_i)-r^x(E_{i-1})$ for each $i \in [k]$.
We now show that $y^*$ belongs to $\conv(\cF)^x$. 
From Proposition~\ref{prop:polymatroid-reduction}, it suffices to show that $y^*(X)\leq r^x(X)$ holds for any $X\subseteq E$. We use an induction on $j\coloneqq\max\set{i|e_i\in X}$. The claim is trivial for $j=1$. For $j\geq 2$, we have $y^*(X)=y^*(X-e_j)+y^*(e_j)\leq r^x(X-e_j)+y^*(e_j)=r^x(X-e_j)+r^x(E_j)-r^x(E_{j-1})\leq r^x(X)$, where the first inequality follows from the induction hypothesis and the last one from submodularity.
Thus, $y^*\in \conv(\cF)^x$ is shown.

Note that $y^*(E_i)=r^x(E_i)$ holds for $i \in [k]$.
This together with the fact that any $y\in \conv(\cF)^x$ satisfies $y(E_i)\leq r^x(E_i)$ for $i \in [k]$ implies that $y^*$ is lexicographically maximum with respect to $\succ$ in $\conv(\cF)^x$ and stochastically dominates all other vectors in $\conv(\cF)^x$. 
\end{proof}

For a matroid $(E, \cF)$ and a total order $\succ$ on $E$, let $F\colon \mathbb{R}_+^E\to\mathbb{R}_+^E$ be a function that returns the vector $F[x]$ that is lexicographically maximum with respect to $\succ$ in $\conv(\cF)^x$ for any given vector $x\in \mathbb{R}_+^E$.
We refer to this function $F$ as the {\em choice function} induced from $(E, \cF)$ and $\succ$.
For example, if $E=\{e_1,e_2,e_3,e_4\}$, $\cF=\{X\subseteq E\mid |X\cap\{e_1,e_3\}|\le 1~\text{and}~|X|\le 2\}$, $e_1\succ e_2\succ e_3\succ e_4$, and $x=(x_{e_1},x_{e_2},x_{e_3},x_{e_4})=(0.4,0.8,1,1)$, the induced choice $F[x]$ is $(0.4,0.8,0.6,0.2)$.
The following fact is an immediate consequence of Lemma~\ref{lem:lexopt=sdopt}.
\begin{lemma}\label{lem:choice-sd}
For any matroid $(E, \cF)$ and any total order $\succ$ on $E$, let $F\colon \mathbb{R}_+^E\to\mathbb{R}_+^E$ be the choice function induced from $(E, \cF)$ and $\succ$. For any $x, y\in \mathbb{R}_+^E$, the condition $x_e\geq y_e~(\forall e\in E)$ implies $F[x]\succeq^\sd F[y]$.
\end{lemma}
\begin{proof}
By Lemma~\ref{lem:lexopt=sdopt}, $F[x]\succeq^\sd z$ for every $z\in \conv(\cF)^x$. 
Because $F[y]\in \conv(\cF)^y$ and $x\geq y$ imply $F[y]\in \conv(\cF)^x$, we obtain $F[x]\succeq^\sd F[y]$.
\end{proof}


Recall that $P\coloneqq \conv(\cA)\subseteq \mathbb{R}^{N\times E}$ denotes the polytope corresponding to the set of feasible fractional assignments.
When $(E, \cF_i)$ is a matroid for every $i\in N$, the following properties are known to hold for $P$: 
(i) $P$ is represented as
\begin{align}
    P=\left\{\pi\in\mathbb{R}^{N\times E} \mathrel{}:\mathrel{} 
    {\begin{array}{l}
    \pi_i\in \conv(\cF_i) \quad(\forall i\in N),\\[1pt]
    \sum_{i\in N}\pi_{ie}\le 1 \quad(\forall e\in E)
    \end{array}}
    \right\} \label{eq:ra-polytope}
\end{align}
\cite{Schrijver2003};
(ii) For a given feasible fractional assignment $\pi\in P$, we can compute in polynomial time a lottery assignment that induces $\pi$ \cite{GLS2012}.
Note that property (i) can be viewed as a generalization of the Birkhoff--von Neumann theorem.

Finally, we state a sufficient condition for a lottery assignment to be sd-envy-free on matroid constraints.
Let $F_i$ be the choice function induced from $(E,\cF_i)$ and $\succ_i$ for each $i\in N$.

\begin{proposition}\label{prop:ef}
Suppose that the constraints are matroid. For any lottery assignment $p\in\Delta(\cA)$, if $\pi_i^p\succeq^\sd_i F_i[\pi_j^p]$ for every $i,j\in N$ then $p$ is sd-envy-free.
\end{proposition}
\begin{proof}
Apply Lemma~\ref{lem:lexopt=sdopt} with $x\coloneqq \pi^p_j$. Then, 
$F_i[\pi^p_j]\succeq_i^\sd y$ holds for every $y\in \conv(\cF_i)^x$.
For each $e\in E$, define $i$'s fractional assignment
$y^{(e)}=\sum_{\bA\in\cA}p_{\bA}\cdot g(\bA,i,j,e)$, where $g(\bA,i,j,e)$ is the characteristic vector of the set $\argmax\{\,|Y\cap U({\succ_i}, e)|\mid  Y\subseteq A_j,\ Y\in \cF_i\,\}$.
Since $y^{(e)}\in \conv(\cF_i)^x$, we have
 $F_i[\pi^p_j]\succeq_i^\sd y^{(e)}$, and hence
\begin{align*}
    F_i[\pi^p_j]&(U(\succ_i,e))\geq y^{(e)}(U(\succ_i,e))
    =\textstyle\sum_{\bA\in\cA}p_{\bA}\max_{Y\subseteq A_j:\,Y\in\cF_i}|Y\cap U(\succ_i, e)|.
\end{align*}
Thus, we can conclude that the statement holds. 
\end{proof}

By summarizing the discussion in this section, we conclude that an sd-envy-free and sd-efficient lottery assignment can be found by computing a feasible fractional assignment $p\in P$ that satisfies $\pi_i^p\succeq^\sd_i F_i[\pi_j^p]~(\forall i,j\in N)$ and $\nexists q\in P\setminus\{p\}$ such that $\pi_i^q\succeq^\sd_i\pi^p_i~(\forall i\in N)$,
when the constraints are matroid.

\section{Tractability Results}
Now, we provide our tractability results.

\subsection{Two agents with matroid constraints}
First, we provide a polynomial-time algorithm to find an sd-efficient and sd-envy-free lottery assignment for the case where there are two agents (i.e., $N=\{1,2\}$) and the constraints $\cF_1$ and $\cF_2$ are matroids. 
Note that due to Example~\ref{ex:PSisBad}, we need a different approach from generalizing the PS mechanism.

Let $(w_{ie})_{i\in N, e\in E}$ be positive weights such that $w_{ia}>w_{ib}$ if and only if $a\succ_i b$ (e.g., $w_{ie}=|\{e'\in E\mid e\succeq_i e'\}|$ for each $i\in N$ and $e\in E$).
Then, a feasible fractional assignment $x\in P$ that maximizes $\sum_{i\in N}\sum_{e\in E}w_{ie}x_{ie}$ is sd-efficient. 
Indeed, if there exists a feasible fractional assignment $x'\in P\setminus\{x\}$ such that $x'_i\succeq_i^\sd x_i$, then we have $\sum_{i\in N}\sum_{e\in E}w_{ie}x'_{ie}>\sum_{i\in N}\sum_{e\in E}w_{ie}x_{ie}$, and hence $x$ does not attain the maximum weight.

This may raise the expectation that a fractional assignment that satisfies both sd-efficiency and sd-envy-freeness can be found by computing a maximum weight feasible fractional assignment subject to an sd-envy-free constraint. If $\cF_1=\cF_2$, such an optimization problem is formulated as the following linear programming:
\begin{align}
    \begin{array}{rl}
         \max       &  \sum_{i\in\{1,2\}}\sum_{e\in E}w_{ie}x_{ie}\\[2mm]
         \text{s.t.}& x_1\succeq_1^\sd x_2,\\ 
                    & x_2\succeq_2^\sd x_1,\\ 
                    & x\in P.
    \end{array}\label{eq:lp}
\end{align}
However, 
the optimal solution for \eqref{eq:lp}
is not always sd-efficient.
To observe this, let us consider an instance where 
$E=\{e_1,e_2,e_3\}$,
$\cF_1=\cF_2=\{X\subseteq E\mid |X\cap \{e_1,e_2\}|\le 1\}$, 
$e_3\succ_1 e_2\succ_1 e_1$, and $e_2\succ_2 e_1\succ_2 e_3$.
By setting $w_{ie}=|\{e'\in E\mid e\succeq_i e'\}|$ for each $i\in N$ and $e\in E$, the optimal solution of the linear programming \eqref{eq:lp} is $(x_1,x_2)=((0,0,1),(0,1,0))$. However, this is not sd-efficient because, for $(y_1,y_2)=((1,0,1),(0,1,0))\in P$, we have $y_1\succ_1^\sd x_1$ and $y_2=x_2$.

In our approach, the following notion will be useful.
\begin{definition}[sd-proportional]
A lottery assignment $p\in \Delta(\cA)$ is called \emph{sd-proportional} if 
$\pi^p_i\succeq_i^\sd x_i$ holds for any $i\in N$ and $x\in P$ with $x_i\leq \frac{1}{n}\cdot\bm{1}$, where $\bm{1}$ is the all-ones vector in $\mathbb{R}^E$.
\end{definition}
We can observe that sd-proportional lottery assignments must exist as follows. Consider a fractional assignment $\pi^*=(\pi^*_1, \pi^*_2)=(F_1[\frac{1}{2}\cdot\bm{1}],F_2[\frac{1}{2}\cdot\bm{1}])$, where $F_i$ 
is the choice function induced from $(E,\cF_i)$ and $\succ_i$ for each $i=1,2$. Then $\pi^*$ belongs to $P=\conv(\cA)$ as $P$ is represented by \eqref{eq:ra-polytope}. Then, there is a lottery assignment that induces $\pi^*$, while $\pi^*$ satisfies the condition for sd-proportionality by Lemma~\ref{lem:lexopt=sdopt}.  
We remark that the existence of sd-proportional lottery assignment does not hold for the general hereditary constraints case (Proposition~\ref{prop:no-prop}).

Furthermore, by using Lemma~\ref{lem:choice-sd}, we can observe that sd-proportionality is a stronger condition than sd-envy-freeness.
\begin{lemma}\label{lem:prop2ef}
If the number of agents is $2$ and the constraints are matroids, sd-proportionality implies sd-envy-freeness.
\end{lemma}
\begin{proof}
Let $x$ be a feasible fractional assignment that satisfies sd-proportionality, that is, $x_i\succeq^\sd_i F_i[\frac{1}{2}\cdot\bm{1}]$ for each $i=1,2$.
It is sufficient to prove that $x_i\succeq^\sd_i F_i[\bm{1}-x_i]$. 
In fact, this implies $x_i\succeq^\sd_i F_i[x_j]$ because 
$F_i[\bm{1}-x_i]\succeq^\sd_i F_i[x_j]$ by $\bm{1}-x_i\ge x_j$ and Lemma~\ref{lem:choice-sd} for $\{i,j\}=\{1,2\}$.

Suppose to the contrary that $x_{i}(U(\succ_i,e))<F_i[\bm{1}-x_i](U(\succ_i,e))$ for some $i\in\{1,2\}$ and $e\in E$.
Let $y_i=\frac{1}{2}(x_i+F_i[\bm{1}-x_i])$. Note that $F_i[y_i]=y_i$. 
Because $2y_i=x_i+F_i[\bm{1}-x_i]\le\bm{1}$, we have $F_i[\frac{1}{2}\cdot\bm{1}]\succeq_i^\sd y_i$. Particularly, we have 
$F_i[\frac{1}{2}\cdot\bm{1}](U(\succ_i,e))\ge y_{i}(U(\succ_i,e))$.
Meanwhile, we have 
\begin{align}
    y_{i}(U(\succ_i,e))
    &=\textstyle\sum_{e':\,e'\succeq_i e}\frac{1}{2}(x_{ie'}+F_i[\bm{1}-x_i]_{e'})
    >x_{i}(U(\succ_i,e))
    \ge F_i[\tfrac{1}{2}\cdot \bm{1}](U(\succ_i,e)),
\end{align}
a contradiction.
\end{proof}

Consider the following linear programming, which is obtained from \eqref{eq:lp} by replacing sd-envy-freeness constraint with sd-envy-proportionality constraint:
\begin{align} 
    \begin{array}{rl}
         \max       &  \sum_{i\in\{1,2\}}\sum_{e\in E} w_{ie}x_{ie}\\[2mm]\label{eq:propLP}
         \text{s.t.} & x_i\succeq_i^\sd F_i[\tfrac{1}{2}\cdot\bm{1}] \qquad (i\in\{1,2\}),\\ 
                     & x\in P.
    \end{array}
\end{align}
We can show that, unlike the case of \eqref{eq:lp}, the optimal solution for \eqref{eq:propLP} is sd-efficient. Furthermore, it is sd-proportional, and hence sd-envy-free by Lemma~\ref{lem:prop2ef}. Thus, we obtain the following theorem.

\begin{theorem}\label{thm:2-mat-het-het}
A lottery assignment that satisfies sd-efficiency and sd-envy-freeness always exists and can be computed in polynomial time if the number of agents is $2$ and the constraints are matroids. 
\end{theorem}
\begin{proof}
Consider the linear programming \eqref{eq:propLP}
with weights defined by $w_{ie}=|\{e'\in E\mid e\succeq_i e'\}|$ for each $i\in \{1,2\}$ and $e\in E$.
The ellipsoid method can solve the problem in polynomial time.
Let $x^*$ be an optimal solution to the problem.
By Lemma~\ref{lem:prop2ef}, $x^*$ is an sd-envy-free feasible fractional assignment.

Now, it suffices to prove that $x^*$ is sd-efficient.
To obtain a contradiction, suppose that $x^*$ is not sd-efficient. Namely, there exists a feasible fractional assignment $y~(\ne x^*)$ such that $y_i\succeq_i^\sd x^*_i$ for $i=1,2$.
Then, $y$ is a feasible solution of \eqref{eq:propLP}. Moreover, the objective value $\sum_{i\in\{1,2\}}\sum_{e\in E} w_{ie} y_{ie}$ is larger than $\sum_{i\in\{1,2\}}\sum_{e\in E}w_{ie}x^*_{ie}$. This contradicts the optimality of $x^*$.
\end{proof}

\subsection{Matroid constraints and identical preferences}
Next, we provide a polynomial-time algorithm to find an sd-efficient and sd-envy-free lottery assignment for the case where the preferences are identical, and the constraints are (heterogeneous) matroids.
Suppose that $E=\{e_1,e_2,\dots,e_m\}$ and 
the preference of each agent $i$ satisfies $e_1\succ_i e_2\succ_i \cdots \succ_i e_m$ without losing generality.
We use $\succ$ to represent $\succ_i$ for simplicity.

Recall that the natural generalization of the PS mechanism does not work for this setting, as shown in Example~\ref{ex:PSisBad}.
We generalize the mechanism in a different way. 
In our algorithm, we regard each item as a divisible item of probability shares and process the items one at a time. 
During the first round, each agent ``eats'' $e_1$ at the same speed while $e_1$ is not eaten up and is available for her.
Similarly, at the $k$th round, each agent eats $e_k$ at the same speed while it remains and is available for her.
Our algorithm is formally described in Algorithm~\ref{alg:psip}, where $\chi_{e_k}$ is the characteristic vector in $\{0,1\}^E$, i.e., its component corresponding to $e\in E$ is $1$ if $e=e_k$ and $0$ otherwise.
Note that Algorithm~\ref{alg:psip} can be implemented to run in polynomial time because $\epsilon_i$ at line 4 can be computed via submodular function minimization (for details, see the proof of Theorem~\ref{thm:n-mat-het-id}).

\begin{algorithm}
\SetInd{.5em}{.5em}
\caption{Heterogeneous matroid constraints and identical preferences}\label{alg:psip}
Let $\bm{x}\ot\bm{0}\in\mathbb{R}^{N\times E}$\;
\For{$k\ot 1,2,\dots,m$}{
\While{True}{
    Let $\epsilon_i\ot\max\{\epsilon\mid x_i+\epsilon\cdot\chi_{e_k}\in \conv(\cF_i)\}$ $(\forall i\in N$)\;\label{line:psip-sat}
    Let $N^+\ot\{i\in N\mid \epsilon_i>0\}$\;
    Let $s\ot \sum_{i\in N}x_{ie_k}$\;
    \lIf{$N^+=\emptyset$ or $s=1$}{\textbf{Break}}
    Let $\epsilon^*\ot \min\{\min_{i\in N^+}\epsilon_i,\, (1-s)/|N^+|\}$\;
    Update $x_i\ot x_i+\epsilon^*\cdot\chi_{e_k}$ for each $i\in N^+$\;
  }
}
\Return a lottery assignment $p\in\Delta(\cA)$ s.t.\ $\pi^p=\bm{x}$\;
\end{algorithm}

For the instance in Example~\ref{ex:PSisBad}, Algorithm~\ref{alg:psip} outputs an sd-efficient and sd-envy-free lottery assignment $p$ such that
\begin{align}
\pi^p = {
\begin{pNiceMatrix}[first-row,first-col]
& e_1 & e_2 & e_3 & e_4 & e_5 \\
1  & 1/2 & 1/2 & 1 & 1/2 & 0\\
2  & 1/2 & 1/2 & 0 & 1/2 & 1
\end{pNiceMatrix}}.
\end{align}


By using properties of matroids shown in Lemma~\ref{lem:choice-sd} and Proposition~\ref{prop:ef}, we can show the correctness of Algorithm~\ref{alg:psip} for an arbitrary number of agents.
\begin{theorem}\label{thm:n-mat-het-id}
An sd-efficient and sd-envy-free lottery assignment always exists and can be computed in polynomial time if the constraints are matroids, and the preferences are identical. 
\end{theorem}
\begin{proof}
We see that
$\epsilon_i=\max\{\epsilon\mid x_i+\epsilon\cdot\chi_{e_k}\in \conv(\cF_i)\}$ at line~\ref{line:psip-sat} of Algorithm~\ref{alg:psip} is equal to 
\begin{align}
\MoveEqLeft
\max\{\epsilon\mid x_i+\epsilon\cdot\chi_{e_k}\in \conv(\cF_i)\}\\
&=\max\{\epsilon\mid x_i(X)+\epsilon\le r_i(X)~(\forall X\subseteq E~\text{such that}~e_k\in X)\}\\
&=\min\{r_i(X)-x_i(X) \mid e_k\in X\subseteq E\}\\
&=\min\bigl\{r_i(Y\cup\{e_k\})-x_i(Y\cup\{e_k\}) \mid Y\subseteq E\setminus\{e_k\}\bigr\},
\end{align}
and the final minimization problem can be solved in polynomial time~\cite{cunningham1984}.
Hence, Algorithm~\ref{alg:psip} can be implemented to run in polynomial time.

It then suffices to prove that the outcome of Algorithm~\ref{alg:psip} satisfies sd-efficiency and sd-envy-freeness.
Let $\bm{x}$ be the fractional assignment obtained at the end of Algorithm~\ref{alg:psip}.
We prove sd-efficiency and sd-envy-freeness of $\bm{x}$ by contradiction.

\paragraph{sd-efficiency}
Suppose to the contrary that there exists a feasible fractional assignment $\bm{y}~(\ne \bm{x})$ such that $y_i\succeq_i^\sd x_i$ for all $i\in N$.
Let $\ell\in[m]$ be the smallest index that satisfies $y_{ie_\ell}>x_{ie_\ell}$ for some $i\in N$.
Let $i^*\in N$ be such an agent $i$ with smallest $x_{ie_\ell}$.
Then, we have $y_{ie_t}=x_{ie_t}$ for all $t\in [\ell-1]$ and $y_{je_\ell}\ge x_{je_\ell}$ for all $j\in N$.
Hence, $i^*$ must have eaten more of $e_\ell$, a contradiction.

\paragraph{sd-envy-freeness}
We will show $x_i \succeq^\sd_i F_i[x_j]$ for any $i,j \in N$. This and Proposition~\ref{prop:ef} imply sd-envy-freeness of $p$.
Suppose that there exists $i,j \in N$ and $e\in E$ such that $x_i(U({\succ},e))<F_i[x_j](U({\succ},e))$.
Let 
\begin{align}
z=F_i\big[(\max\{x_{ie},\,F_i[x_j]_e\})_{e\in E}\big]\in\conv(\cF_i).
\end{align}
Note that $x_i(U({\succ},e))<F_i[x_j](U({\succ},e))\le z(U({\succ},e))$ by Lemma~\ref{lem:choice-sd}.
Let $\ell\in[m]$ be the smallest index that satisfies $x_i(U({\succ},e_\ell))<z(U({\succ},e_\ell))$.
We have $x_{ie_t}=z_{e_t}$ for all $t\in[\ell-1]$ because, for each $t\in[\ell-1]$,
$x_i(U({\succ},e_t))\ge z(U({\succ},e_t))$ by the choice of $\ell$ and 
$x_i(U({\succ},e_t))\le z(U({\succ},e_t))$ by Lemma~\ref{lem:choice-sd}.
Hence, $x_{ie_\ell}<z_{e_{\ell}}\le \max\{x_{ie_\ell},F_i[x_j]_{e_\ell}\} = F_i[x_j]_{e_\ell}\le x_{je_\ell}$.
Thus, when $i$ has eaten $x_{ie_\ell}$ units of $e_\ell$ in the algorithm, 
$e_\ell$ has not been eaten up, particularly $j$ is still eating $e_\ell$ because $x_{je_\ell}>x_{ie_\ell}$.  
Additionally, $i$ can eat $e_\ell$ more while preserving her feasibility because $x_{ie_t}=z_{e_t}$ for all $t\in[\ell-1]$, $x_{ie_\ell}<z_{e_\ell}$, and $z\in\conv(\cF_i)$.
This contradicts the behavior of the algorithm.
\end{proof}

\subsection{Identical constraints and preferences}
Finally, we provide the existence result when the constraints and preferences are identical.

\begin{theorem}\label{thm:n-gen-id-id}
An sd-efficient and sd-envy-free lottery assignment must exist for any instance with identical constraints and preferences. 
\end{theorem}
\begin{proof}
Let $\bA^*\in\cA$ be a deterministic assignment such that the set of assigned items $\bigcup_{i\in N}A^*_i$ is lexicographically maximum with respect to the common preference. Let $\bA^{(1)}\coloneqq \bA^*$, then, for each $i$ with $1<i\leq n$, let $\bA^{(i)}$ be obtained by shifting the bundles of $\bA^*$ by $i-1$, that is, $\bA^{(i)}\coloneqq (A^*_i,\ldots,A^*_n,A^*_1,\ldots,A^*_{i-1})$.
Define the lottery assignment $p\in\Delta(\cA)$ by setting $p_{\bA^{(i)}}=1/n$ for each $i\in [n]$. Then, $p$ is sd-envy-free because $\pi^p_1=\dots=\pi^p_n$. 
To observe that $p$ is sd-efficient, suppose to the contrary that there exists a lottery assignment $q\in\Delta(\cA)$ with $q\neq p$, such that $\pi_i^q \succeq^\sd \pi_i^p$ for all $i\in N$ with respect to the common preference. 
We define vectors $v^p,v^q\in \mathbb{R}^E$ by $v^p_e=\sum_{i\in N}\pi_{ie}^p$ and $v^q_e=\sum_{i\in N}\pi_{ie}^q$ for each $e\in E$. That is, $v^p_e$, $v^q_e$ are the probabilities that item $e$ is allocated to someone under $p$ and $q$, respectively.
Since $\pi_i^q \succeq^\sd \pi_i^p$ for all $i\in N$, we obtain $v^q \succeq^\sd v^p$. 
Because $v^p$ coincides with the characteristic vector of $\bigcup_{i\in N}A^*_i$, 
this implies that the support of $q$ contains a deterministic assignment $\bm{Y}=(Y_1, Y_2,\dots,Y_n)\in \cA$ such that $\bigcup_{i\in N}Y_i$ is lexicographically larger than $\bigcup_{i\in N}A^*_i$. This contradicts the choice of $\bA^*$.
\end{proof}

We note that an sd-efficient and sd-envy-free lottery assignment can be computed in polynomial time when the agents have identical matroid constraints and identical preferences.
In contrast, for general identical constraints, computing such a lottery assignment is NP-hard even if $n=2$ (see Theorem~\ref{thm:NPhard}).

\section{Impossibility Results}
In this section, we present impossibility results.

\subsection{Identical preferences}
We first consider the case where the preferences are identical.
We have demonstrated that, if the constraints are matroids, then an sd-efficient and sd-envy-free lottery assignment must exist (Theorem~\ref{thm:n-mat-het-id}). However, this is not true for general hereditary constraints.

\begin{theorem}\label{thm:2-gen-het-id}
An sd-efficient and sd-envy-free lottery assignment may not exist even with two agents, and the preferences are identical.
\end{theorem}
\begin{proof}
Let $(N,E,(\succ_i, \cF_i)_{i\in N})$ be an instance where
$N=\{1,2\}$, $E=\{e_1,e_2,e_3,e_4\}$, $\cF_1=2^{\{e_1\}}\cup 2^{\{e_2\}}\cup 2^{\{e_3,e_4\}}$, $\cF_2=2^E$, and $e_1\succ_i e_2 \succ_i e_3 \succ_i e_4\ (i=1,2)$.
Then, we show that no lottery assignment in this instance satisfies sd-efficiency and sd-envy-freeness simultaneously.

Suppose to the contrary that $p\in\Delta(\cA)$ is an sd-efficient and sd-envy-free lottery assignment.
Because $\cF_2$ is $2^E$ and $p$ is sd-efficient, we have $p_{\bA}>0$ only if $\bA=(A_1,A_2)$ satisfies $A_1\cup A_2=E$.
Thus, $\pi^p_{1e}+\pi^p_{2e}=1$ for all $e\in E$.
By the sd-envy-freeness of $p$, we must have $\pi^p_{1e_1}=\pi^p_{2e_1}=1/2$, and hence $p_{(\{e_1\},\{e_2,e_3,e_4\})}=1/2$.
Additionally, $\pi^p_{1e_2}=\pi^p_{2e_2}=1/2$ by the sd-envy-freeness, and hence $p_{(\{e_2\},\{e_1,e_3,e_4\})}=1/2$.
Then, the condition \eqref{eq:sdEF} is violated for $i=1$, $j=2$, and $e=e_4$ because we have
\begin{align}
\sum_{\bA\in\cA}p_{\bA}|A_1\cap U(\succ_1, e_4)|
<\sum_{\bA\in\cA}p_{\bA}\max_{\substack{\scriptscriptstyle Y\subseteq A_2: \\\scriptscriptstyle Y\in\cF_1}}|Y\cap U(\succ_1, e_4)|,
\end{align}
where the left-hand side is $1$ and the right-hand side is $2$.
This contradicts sd-envy-freeness of $p$.
\end{proof}

\subsection{Identical matroid constraints}
Next, we observe the case where the constraints are an identical matroid. 
\begin{theorem}\label{thm:3-mat-id-het}
An sd-efficient and sd-envy-free lottery assignment may not exist even with three agents and identical matroid constraints.
\end{theorem}
\begin{proof}
Let $(N,E,(\succ_i, \cF_i)_{i\in N})$ be an instance where
$N=\{1,2,3\}$, $E=\{a,b,c,d,e\}$, $\cF_i=\{X\subseteq E\mid |X\cap\{a,b,c\}|\le 1\}$ for all $i\in N$, and 
\begin{align}
    &d\succ_1 a \succ_1 b \succ_1 c \succ_1 e,\\
    &d\succ_2 b \succ_2 e \succ_2 a \succ_2 c,\\
    &a\succ_3 d \succ_3 e \succ_3 b \succ_3 c.
\end{align}
We prove that no lottery assignment in this instance satisfies sd-efficiency and sd-envy-freeness simultaneously.

To obtain a contradiction, suppose that an sd-efficient and sd-envy-free lottery assignment induces $\pi\in P$.
Then, by sd-efficiency, each agent receives a unit amount of item $\{a,b,c\}$ and $2/3$ amount of $\{d,e\}$.
Let $\pi_{1d}=\alpha$, $\pi_{3a}=\beta$, and $\pi_{2c}=\gamma$.
By sd-envy-freeness, the fractional assignment $\pi$ can be written as follows:
\[
{\begin{pNiceMatrix}[first-row,first-col]
& a & b & c & d & e\\
1  & 1{-}3\alpha{+}\beta & {-}1{+}3\alpha{-}\beta{+}2\gamma & 1{-}2\gamma & \alpha    &  \frac{2}{3}-\alpha\\[2pt]
2  & 3\alpha{-}2\beta  &  1{-}3\alpha{+}2\beta{-}\gamma & \gamma    & \alpha    &  \frac{2}{3}{-}\alpha\\[2pt]
3  & \beta           &  1{-}\beta{-}\gamma           & \gamma    & 1{-}2\alpha & {-}\frac{1}{3}{+}2\alpha\\
\end{pNiceMatrix}}.
\]
As $\pi$ is a feasible fractional assignment, we have $\pi_{3d}=1-2\alpha\ge 0$ and $\pi_{2a}=3\alpha-2\beta\ge 0$.
Thus, we obtain
\begin{align}
\alpha\le 1/2 \quad\text{and}\quad \beta\le 3\alpha/2\le 3/4. \label{eq:alpha_beta_upper}
\end{align}
Moreover, we have $\alpha+2\gamma=\pi_{1d}+\pi_{1a}+\pi_{1b}\ge \pi_{2d}+\pi_{2a}+\pi_{2b}=1+\alpha-\gamma$ by sd-envy-freeness; therefore,
\begin{align}
\gamma\ge 1/3. \label{eq:gamma_lower}
\end{align}
For a sufficiently small positive $\epsilon$, let 
\begin{align}
\pi' &\coloneqq ~\pi +
{\begin{pNiceMatrix}[first-row,first-col]
& a & b & c & d & e\\
1  & \epsilon  & -\epsilon & 0 & 0 & 0\\
2  & -\epsilon &  \epsilon & 0 & 0 & 0\\
3  & 0         & 0         & 0 & 0 & 0\\
\end{pNiceMatrix}}
\quad\text{and}\quad
\pi'' \coloneqq ~\pi + 
{\begin{pNiceMatrix}[first-row,first-col]
& a & b & c & d & e\\
1  & -\epsilon &  0 & \epsilon  & \epsilon  &  -\epsilon\\
2  & 0         &  0 & 0         & 0         &  0\\
3  & \epsilon  &  0 & -\epsilon & -\epsilon & \epsilon\\
\end{pNiceMatrix}}.
\end{align}
Because $\pi'$ and $\pi''$ improve $\pi$, they must be infeasible.
By the infeasibility of $\pi'$, we have (i) $\pi_{1b}=-1+3\alpha-\beta+2\gamma=0$ or (ii) $\pi_{2a}=3\alpha-2\beta=0$ (because $\pi_{1b}>0$ implies $\pi_{2b}<1$ and $\pi_{2a}>0$ implies $\pi_{1a}<1$).
Additionally, by the infeasibility of $\pi''$, we have (iii) $\pi_{1a}=1-3\alpha+\beta=0$ or (iv) $\pi_{3d}=1-2\alpha=0$ 
(because $\pi_{1c}=1-2\gamma\le 1/3<1$, $\pi_{1d}=\alpha\le 1/2<1$, $\pi_{1e}=2/3-\alpha\ge 1/6>0$, $\pi_{3a}=\beta\le 3/4<1$, $\pi_{3c}=\gamma\ge 1/3>0$, and $\pi_{3e}=-1/3+2\alpha\le 2/3<1$ from \eqref{eq:alpha_beta_upper} and \eqref{eq:gamma_lower}).
We consider four possible cases separately.

\medskip
\noindent\textbf{Case 1} (i) $-1+3\alpha-\beta+2\gamma=0$ and (iii) $1-3\alpha+\beta=0$.
Then, $\gamma=0$, which contradicts $\gamma\ge 1/3$ from \eqref{eq:gamma_lower}.

\medskip
\noindent\textbf{Case 2} (i) $-1+3\alpha-\beta+2\gamma=0$ and (iv) $1-2\alpha=0$.
Then, $\gamma=(1-3\alpha+\beta)/2=(-1/2+\beta)/2\ge 1/3$ from \eqref{eq:gamma_lower}.
This implies $\beta\ge 7/6$, which contradicts $\beta=\pi_{3a}\le 1$.

\medskip
\noindent\textbf{Case 3} (ii) $3\alpha-2\beta=0$ and (iii) $1-3\alpha+\beta=0$.
Then, $\alpha=2/3$, which contradicts $\alpha\le 1/2$ from \eqref{eq:alpha_beta_upper}.

\smallskip
\noindent\textbf{Case 4} (ii) $3\alpha-2\beta=0$ and (iv) $1-2\alpha=0$.
Then, $\alpha=1/2$ and $\beta=3/4$.
Hence, $\pi_{3b}=1/4-\gamma\ge 0$, which contradicts $\gamma\ge 1/3$ from \eqref{eq:gamma_lower}.
\smallskip

Thus, no fractional assignment in the instance satisfies sd-efficiency and sd-envy-freeness simultaneously.
\end{proof}

We can also demonstrate that, for any $n\ge 3$, there exists an instance that has no sd-efficency and sd-envy-freeness lottery assignment.
Consider an instance $(N,E,(\succ_i,\cF_i)_{i\in N})$ where $N=\{1,2,\dots,n\}$, $E=\{a,b,c,d,e,o_6,\dots,o_{2n}\}$, $\cF_i=\{X\subseteq E\mid |X\cap\{a,b,c\}|\le 1\}$ for all $i\in N$, and 
\begin{align}
    &d\succ_1 a \succ_1 b \succ_1 c \succ_1 e \succ_1 o_6\succ_1\cdots,\\
    &d\succ_2 b \succ_2 e \succ_2 a \succ_2 c \succ_2 o_6\succ_2\cdots,\\
    &a\succ_3 d \succ_3 e \succ_3 b \succ_3 c \succ_3 o_6\succ_3\cdots,\\
    &o_{2i-1} \succ_i o_{2i} \succ_i \cdots \quad(i=4,5,\dots,n).
\end{align}
Then, analysis similar to that in the proof of Theorem~\ref{thm:3-mat-id-het} demonstrates that this instance does not have an sd-efficient and sd-envy-free lottery assignment.

\section{Computational Results}
This section discusses computational issues related to finding an sd-efficient and sd-envy-free lottery assignment.
Of course, finding such a lottery assignment is NP-hard, even for a simple case.
It is worth noting that this problem remains challenging even with exponential time, given the intricate nature of sd-efficiency representation.
Nevertheless, we provide an exponential-time algorithm to solve it. 

\subsection{Computational hardness}
Throughout this subsection, we assume that the input of each constraint is given as a system of linear inequalities.
We first prove that checking the sd-efficiency of a given assignment is NP-hard even when there is only one agent.
\begin{theorem}\label{thm:checkHard}
Even when there is only one agent, checking the sd-efficiency of a given assignment is coNP-hard.
\end{theorem}
\begin{proof}
We provide a reduction from the \emph{3-dimensional matching (3DM)} problem, which is known to be NP-hard~\cite{GJ1979}. 
An instance of the problem consists of $3$ sets $X=\{x_1,x_2,\dots,x_n\}$, $Y=\{y_1,y_2,\dots,y_n\}$, $Z=\{z_1,z_2,\dots,z_n\}$ and a subset $T=\{t_1,t_2,\dots,t_m\}$ of $X\times Y\times Z$.
A subset $M\subseteq T$ is called a $3$-dimensional matching if the following holds: for any two distinct triples $t=(x,y,z)\in M$ and $t'=(x',y',z')\in M$, we have $x\ne x'$, $y\ne y'$, and $z\ne z'$.
Our task is to decide whether a $3$-dimensional matching of size $n$ exists.
Given an instance of the 3DM problem $(X,Y,Z;T)$, we construct an instance of our assignment problem $(N,E,(\succ_i,\cF_i)_{i\in N})$ as follows:
\begin{itemize}
\item $N=\{1\}$,
\item $E=T\cup S$ where $S=\{s_1,s_2,\dots,s_n\}$,
\item $t_1\succ_1\dots\succ_1 t_m\succ_1 s_1\succ_1\dots\succ_1 s_n$, and
\item $\cF_1=\{M\subseteq T\mid M\text{ is a $3$-dimensional matching}\}\cup 2^S$.
\end{itemize}
Then, it is not difficult to see that the deterministic assignment where agent $1$ receives $S$ is sd-efficient if and only if the 3DM instance $(X,Y,Z;T)$ is a No-instance.
Therefore, checking the sd-efficiency of a given assignment is coNP-hard.
\end{proof}
It should be noted that, if there is only one agent, an sd-efficient and sd-envy-free lottery assignment must exist and can be computed.
Indeed, one such assignment is to assign the lexicographic maximum set to the agent with probability $1$, which can be computed in polynomial time.
In addition, when there are multiple agents but have identical constraints, verifying the sd-envy-freeness of a given lottery assignment $p$ can be done efficiently by checking $\pi^p_i\succeq_i^\sd \pi_j^p$ for all pairs of agents $i$ and $j$ in $N$.

Next, we demonstrate that finding an sd-efficient and sd-envy-free lottery assignment is computationally hard even in the case where there are two agents with identical preferences and constraints.
\begin{theorem}\label{thm:NPhard}
Even when there are two agents with identical preferences and constraints, finding an sd-efficient and sd-envy-free lottery assignment is NP-hard.
\end{theorem}
\begin{proof}
We reduce from the \emph{PARTITION} problem, which is known to be NP-hard~\cite{GJ1979}. 
An instance of the problem consists of $k$ integers $a_1,a_2,\dots,a_k$, and it is a Yes-instance if and only if there exists $I\subseteq [k]$ such that $\sum_{i\in I}a_i=\sum_{i\in [k]\setminus I}a_i$.
Given an instance of the PARTITION problem $(a_1,a_2,\dots,a_n)$, we construct an instance of our assignment problem $(N,E,(\succ_i,\cF_i)_{i\in N})$ as follows:
\begin{itemize}
\item $N=\{1,2\}$, 
\item $E=\{e_1,e_2,\dots,e_k\}$, 
\item $e_1\succ_ie_2\succ_i\dots\succ_i e_k~(i\in \{1,2\})$, and 
\item $\cF_1=\cF_2=\{E'\subseteq E\mid \sum_{e_i\in E'}a_i\le \sum_{i\in[k]}a_i/2\}$.
\end{itemize}
If the instance of the PARTITION problem is a Yes-instance, an sd-efficient and sd-envy-free lottery assignment $p\in\Delta(\cA)$ must satisfy $\pi^p_1=\pi^p_2=\frac{1}{2}\cdot\bm{1}$.
On the other hand, if the instance of the PARTITION problem is a No-instance, any feasible assignment $p\in\Delta(\cA)$ cannot satisfy $\pi^p_1=\pi^p_2=\frac{1}{2}\cdot\bm{1}$.
Hence, finding an sd-efficient and sd-envy-free lottery assignment is NP-hard.
\end{proof}

\subsection{An exponential-time algorithm}
Here, we provide an algorithm that finds an sd-efficient and sd-envy-free lottery assignment in an exponential time with respect to the numbers of agents and items.
Let $(N,E,(\succ_i,\cF_i)_{i\in N})$ be an instance of our assignment problem, where $N=\{1,2,\dots,n\}$ and $E=\{e_1,e_2,\dots,e_m\}$. Define $\cA$ as the set of all deterministic assignments.
As a preparation for constructing the algorithm, we provide a characterization of the support of sd-efficient lottery assignments.
\begin{lemma}\label{lem:efficient_base}
If a lottery assignment $p\in\Delta(\cA)$ is sd-efficient, then every lottery assignment $q\in\Delta(\cA)$ such that $q_{\bA}=0$ for all $\bA\in\cA$ with $p_{\bA}=0$ is also sd-efficient.
\end{lemma}
\begin{proof}
Let $p$ and $q$ be lottery assignments such that $\{\bA\in\cA\mid p_{\bA}>0\}\supseteq\{\bA\in\cA\mid q_{\bA}>0\}$. 
Suppose to the contrary that $p$ is sd-efficient, but $q$ is not sd-efficient. 
Then, there exists a lottery assignment $q'\in\Delta(\cA)$ such that $\pi_i^{q'}\succeq_i^\sd \pi_i^q$ for every $i\in N$ and $\pi_i^{q'}\succ_i^\sd \pi_i^q$ for some $i\in N$.
Let $\epsilon$ be a positive real such that $\epsilon<p_{\bA}$ for all $\bA\in\cA$ with $p_{\bA}>0$. 
We construct a lottery assignment $p'\in\Delta(\cA)$ by setting $p'_{\bA}=p_{\bA}+\epsilon\cdot(q'_{\bA}-q_{\bA})$ for all $\bA\in\cA$.
Here, $p'$ is indeed a lottery assignment because 
$\sum_{\bA\in\cA}p'_{\bA}=\sum_{\bA\in\cA}[p_{\bA}+\epsilon\cdot(q'_{\bA}-q_{\bA})]=1+\epsilon(1-1)=1$, 
$p'_{\bA}\ge p_{\bA}-\epsilon>0$ for each $\bA\in\cA$ with $p_{\bA}>0$, and 
$p'_{\bA}=\epsilon\cdot q'_{\bA}\ge 0$ for each $\bA\in\cA$ with $p_{\bA}=0$ by $q_{\bA}=0$.
Then, we have $\pi_i^{p'}\succeq_i^\sd \pi_i^p$ for every $i\in N$ and $\pi_i^{p'}\succ_i^\sd \pi_i^p$ for some $i\in N$.
This means that $p$ is not sd-efficient, leading to a contradiction.
\end{proof}

For a subset of deterministic assignments $\cA' \subseteq \cA$, let $\Delta(\cA')$ be the set of lottery assignments $p \in \Delta(\cA)$ such that $p_{\bA} = 0$ for all $\bA \not \in \cA'$.
Let $\sA \subseteq 2^{\cA}$ be the family of subsets of deterministic assignments $\cA' \subseteq \cA$ such that all the lottery assignments in $\Delta(\cA')$ are sd-efficient.
Then, by Lemma~\ref{lem:efficient_base}, the pair $(\cA, \sA)$ forms an independence system. Here, for each $\cA' \subseteq \cA$, we can check whether $\cA' \in \sA$ or not by verifying the sd-efficiency of the lottery assignment $u \in \Delta(\cA')$ defined by $u_{\bA}=1/|\cA'|$ if $\bA \in \cA'$ and $u_{\bA}=0$ otherwise.
Indeed, this can be done by checking that the optimal value of the following LP is zero:
\begin{align} 
    \begin{array}{rll}
         \max       &  \sum_{i\in N}\sum_{e\in E} s_{ie} &\\[2mm]\label{eq:effLP}
         \text{s.t.} & \sum_{e'\in U(\succ_i,e)}(x_{ie'}-\pi_{ie'}^{u}-s_{ie'})\ge 0    &(\forall i\in N,\, e\in E),\\ 
                     & x_{ie}=\sum_{\bA\in\cA:\,e\in A_i}p_{\bA} &(\forall i\in N,\, e\in E),\\
                     & \sum_{\bA\in\cA}p_{\bA}=1,&\\
                     & p_{\bA}\ge 0                              &(\forall \bA\in\cA),\\
                     & x_{ie},s_{ie}\ge 0                               &(\forall i\in N,\,e\in E).
    \end{array}
\end{align}
Moreover, given $\cA'\in\sA$, an sd-envy-free lottery assignment in $\Delta(\cA')$ can be obtained if it exists by solving the following linear inequalities:
\begin{align} 
\begin{array}{ll}\label{eq:efLP}
\sum_{\bA\in\cA'} p_{\bA}|A_i\cap U(\succ_i, e)|\ge \sum_{\bA\in\cA'} p_{\bA}\max_{\substack{\scriptscriptstyle Y\subseteq A_j: \scriptscriptstyle Y\in\cF_i}} |Y\cap U(\succ_i, e)|\quad
\left(\begin{tabular}{c}$\forall i,j\in N$,\\ $\forall e\in E$\end{tabular}\right)\\
\sum_{\bA\in\cA}p_{\bA}=1,\\
p_{\bA}\ge 0  \quad(\forall \bA\in\cA').
\end{array}
\end{align}

The number of possible deterministic assignments $|\cA|$ is at most $(n+1)^m$, where $n$ is the number of agents and $m$ is the number of items.
In addition, by Carath\'eodory's theorem, we only need to consider subsets of $\cA$ with at most $nm+1$ deterministic assignments.
Hence, we only need to determine existence of an sd-envy-free lottery assignment for $(n+1)^{O(nm^2)}$ candidates $\{\cA'\in \sA\mid |\cA'|\le nm+1\}$.
For each $u\in\Delta(\cA)$, the corresponding LP of \eqref{eq:effLP} can be solved in polynomial time with respect to $n$, $m$, and $|\cA|~(\le (n+1)^m)$.
For each $\cA'\in\sA$ with $|\cA'|\le nm+1$, the corresponding linear inequalities of \eqref{eq:efLP} can be solved in polynomial time with respect to $n$ and $m$.
Thus, an sd-efficient and sd-envy-free assignment can be found in an exponential time using Algorithm~\ref{alg:enumerate}.

\begin{algorithm}
\SetInd{.5em}{.5em}
\caption{Identical constraints}\label{alg:enumerate}
Construct the set of all deterministic assignments $\cA$\;
\ForEach{$\cA'\subseteq \cA$ with $|\cA'|\le nm+1$}{
    Check whether $\cA'\in \sA$ by solving \eqref{eq:effLP}\;
    \If{$\cA'\in \sA$}{
        Check the existence an sd-envy-free lottery assignment $p$ in $\Delta(\cA')$ by solving \eqref{eq:efLP}\;
        \lIf{such a lottery assignment $p$ exists }{\Return $p$}
    }
}
\Return ``No sd-efficient and sd-envy-free lottery assignment exists''\;
\end{algorithm}

\begin{theorem}
If an sd-efficient and sd-envy-free lottery assignment exists, 
Algorithm~\ref{alg:enumerate} outputs such a lottery assignment.
\end{theorem}

We have implemented this algorithm and experimented with millions of random instances with $2$ agents with identical constraints and $6$ items. 
In all cases of them, it found sd-efficient and sd-envy-free lottery assignments.
We leave a conjecture that there always exists an sd-efficient and sd-envy-free lottery assignment for two agents with identical constraints.


\section*{Acknowledgments}

We would like to thank the anonymous reviewers for their valuable comments. 
This work was partially supported by 
JSPS KAKENHI Grant Numbers 
JP17K12646, 
JP18K18004, 
JP20K19739, 
JP21K17708, 
and JP21H03397, 
by JST PRESTO Grant Numbers 
JPMJPR2122 and 
JPMJPR212B, 
JST ERATO Grant Number JPMJER2301, Japan,
and by Value Exchange Engineering, a joint research project between Mercari, Inc. 




\appendix

\section{Proof of Proposition~\ref{prop:sdEF}}\label{app:sdEF}
We first prove that \eqref{eq:sdEF0} implies \eqref{eq:sdEF}.
Let us fix $i,j\in N$ and $e\in E$. 
For a positive real $\epsilon$, define $u^{(\epsilon)}_i\colon E\to\mathbb{R}_{++}$ to be the utility function of $i$ such that 
\begin{align}
u^{(\epsilon)}_i(e')=\begin{cases}
1+\epsilon\cdot|E\setminus U(\succ_i,e')|&\text{if }e'\succeq e,\\
\epsilon\cdot|E\setminus U(\succ_i,e')|&\text{otherwise}\end{cases}
\end{align}
for each $e'\in E$. This utility function is consistent to $\succ_i$ if $\epsilon<1/|E|$.
By \eqref{eq:sdEF0}, we have
\begin{align}
\sum_{\bA\in\cA}p_{\bA}|A_i\cap U(\succ_i, e)|
&\ge \sum_{\bA\in\cA}p_{\bA}\sum_{e'\in A_i}(u_i^{(\epsilon)}(e')-\epsilon|E|)\\
&= \sum_{\bA\in\cA}p_{\bA}(\tilde{u}_i^{(\epsilon)}(A_i)-\epsilon|E|^2)\\
&= \mathbb{E}_{\bA\sim p}[\tilde{u}_i^{(\epsilon)}(A_i)]-\epsilon|E|^2\\
&\ge \mathbb{E}_{\bA\sim p}[\tilde{u}_i^{(\epsilon)}(A_j)]-\epsilon|E|^2\\
&= \sum_{\bA\in\cA}p_{\bA}u_i^{(\epsilon)}(A_j)-\epsilon|E|^2\\
&\ge \sum_{\bA\in\cA}p_{\bA}\Bigg(\max_{\substack{Y\subseteq A_j:\\Y\in\cF_i}}|Y\cap U(\succ_i,e)|-\epsilon |E|^2\Bigg) -\epsilon|E|^2\\
&= \sum_{\bA\in\cA}p_{\bA}\max_{\substack{Y\subseteq A_j:\\Y\in\cF_i}}|Y\cap U(\succ_i,e)| - 2\epsilon |E|^2.
\end{align}
By taking the limit as $\epsilon$ goes to zero, we obtain
\begin{align}
\sum_{\bA\in\cA}p_{\bA}|A_i\cap U(\succ_i, e)|
\ge \sum_{\bA\in\cA}p_{\bA}\max_{\substack{Y\subseteq A_j\cap U(\succ_i,e)\\Y\in\cF_i}}|Y|.
\end{align}

Conversely, we now prove that \eqref{eq:sdEF} implies \eqref{eq:sdEF0}.
Let us fix $i,j\in N$ and $u_i$ consistent to $\succ_i$.
Consider a function $C\colon 2^E\to 2^E$ such that $C(X)\in\argmax_{Y\subseteq X:\, Y\in\cF_i}\tilde{u}_i(Y)$.
Then, \eqref{eq:sdEF} implies
\begin{align}
\sum_{\bA\in\cA}p_{\bA}|A_i\cap U(\succ_i, e)|
&\ge \sum_{\bA\in\cA}p_{\bA}\max_{\substack{Y\subseteq A_j\cap U(\succ_i,e)\\Y\in\cF_i}}|Y|\\
&= \sum_{\bA\in\cA}p_{\bA}|C(A_j)\cap U(\succ_i,e)|
\end{align}
for each $e\in E$.
In addition, $\tilde{u}_i(X)=\sum_{e\in C(X)}u_i(e)=\tilde{u}_i(C(X))$.
We relabel the items $E=\{e_1,\dots,e_m\}$ so that $e_1\succ_i\dots\succ_i e_m$.
For notational convenience, we denote $u_i(e_{m+1})=0$.
Now, we get
\begin{align}
\mathbb{E}_{\bA\sim p}[\tilde{u}_i(A_i)]
&=\sum_{\bA\in\cA}p_{\bA}\cdot\tilde{u}_i(A_i)\\
&=\sum_{\bA\in\cA}p_{\bA}\sum_{k=1}^m(u_i(e_k)-u_i(e_{k+1}))|A_i\cap U(\succ_i,e_k)|\\
&=\sum_{k=1}^m(u_i(e_k)-u_i(e_{k+1}))\sum_{\bA\in\cA}p_{\bA}|A_i\cap U(\succ_i,e_k)|\\
&\ge \sum_{k=1}^m(u_i(e_k)-u_i(e_{k+1}))\sum_{\bA\in\cA}p_{\bA}|C(A_j)\cap U(\succ_i,e_k)|\\
&=\sum_{\bA\in\cA}p_{\bA}\sum_{k=1}^m(u_i(e_k)-u_i(e_{k+1}))|C(A_j)\cap U(\succ_i,e_k)|\\
&=\sum_{\bA\in\cA}p_{\bA}\cdot\tilde{u}_i(C(A_j))
=\sum_{\bA\in\cA}p_{\bA}\cdot\tilde{u}_i(A_j)\\
&=\mathbb{E}_{\bA\sim p}[\tilde{u}_i(A_j)].
\end{align}

\section{Compatibility of sd-weak-strategy-proofness}\label{sec:SP}
A \emph{mechanism} $\phi$ is a mapping from a preference profile to a lottery assignment.
A mechanism $\phi$ satisfies \emph{sd-weak-strategy-proofness} if 
for any preference profile $\succ$, agent $i$, and a preference $\succ_i'$, it holds that
\begin{align}
\phi(\succ'_i,\succ_{-i})_i\succeq^\sd_i \phi(\succ)_i \implies \phi(\succ'_i,\succ_{-i})_i=\phi(\succ)_i.
\end{align}

It is known that no mechanism simultaneously satisfies sd-efficiency, sd-envy-freeness, and sd-weak-strategy-proofness, even for $2$ agents with no constraints~\cite{kojima2009}.
However, even for the most general setting\footnote{The number of agents is arbitrary, the constraints are any hereditary constraints, the constraints of the agents can be heterogeneous, and the ordinal preferences of the agents can be heterogeneous.}, it is possible to construct a mechanism by giving up either sd-efficiency or sd-envy-freeness.
In fact, the mechanism that always returns the deterministic assignment of $(\emptyset,\emptyset,\dots,\emptyset)$ is sd-envy-free and sd-weak-strategy-proof.
Moreover, a serial dictatorship mechanism that sequentially assigns the lexicographically maximum feasible set of remaining items to each agent in a fixed agents' order is sd-efficient and sd-weak-strategy-proof.

\section{Other desirable properties}\label{app:others}
\subsection{Existence of anonymous mechanism}
In this section, we consider \emph{anonymity} as a fairness concept.
A mechanism is called \emph{anonymous} if permuting agents' names does not affect the resulting outcome.
Note that anonymity implies \emph{equal treatment of equals}, namely, two agents with identical preferences and constraints get the same fractional assignment.

In general, the set of feasible fractional assignments of $(N,E,(\succ_i)_{i\in N})$ can be written as follows:
\begin{align}
    P = \left\{\pi^p\in \mathbb{R}^{N\times E} \mid p\in\Delta(\cA)\right\}.
\end{align}
Note that $P$ is a nonempty polytope.
Let us consider a mechanism that outputs the optimal solution of the following convex programming:
\begin{align}
    \begin{array}{rl}
         \min       &  \sum_{i\in N}\sum_{e\in E}\left(\sum_{e'\in E:\,e'\succeq_i e}(1-x_{ie'})\right)^2\\
         \text{s.t.}&  x\in P.
    \end{array}
\end{align}
This convex programming has a unique optimal solution because the objective function is strictly convex, and the feasibility region is nonempty, closed, bounded, and convex~\cite{RockWets98}.
Hence, the mechanism is anonymous because the convex programming does not depend on the ordering of agents.

Let $x^*$ be the optimal solution to the problem.
We show that $x^*$ is sd-efficient by contradiction.
Suppose that $x^*$ is not sd-efficient, that is, there exists a feasible fractional assignment $y\in P\setminus\{x^*\}$ such that $y_i\succeq_i^\sd x^*_i$ for all $i\in N$. 
Then the objective value of $x^*$ (i.e., $\sum_{i\in N}\sum_{e\in E}(\sum_{e'\in E:\,e'\succeq_i e}(1-x^*_{ie'}))^2$) is larger than the objective value of $y$ (i.e., $\sum_{i\in N}\sum_{e\in E}(\sum_{e'\in E:\,e'\succeq_i e}(1-y_{ie'}))^2$).
This contradicts the optimality of $x^*$.

\begin{theorem}\label{thm:anonymous}
There exists an anonymous mechanism that always produces an sd-efficient lottery assignment.
\end{theorem}

\subsection{Nonexistence of necessarily Pareto-efficient assignment}

\begin{proposition}\label{prop:no-necessaryPO}
A lottery assignment that satisfies necessarily Pareto-efficient may not exist.
\end{proposition}
\begin{proof}
Consider the following instance with $N=\{1\}$ and $E=\{a,b,c\}$:
\begin{itemize}
    \item $\cF_1=\{\emptyset,\{a\},\{b\},\{c\},\{b,c\}\}$, and
    \item $a\succ_1 b\succ_1 c$.
\end{itemize}
If the cardinal utilities are $u_1(a)=4$, $u_1(b)=2$, and $u_1(c)=1$, Pareto-efficient lottery assignment is uniquely determined as the deterministic assignment $(\{a\})$.
If the cardinal utilities are $u_1(a)=4, u_1(b)=3, u_1(c)=2$, Pareto-efficient lottery assignment is uniquely determined as the deterministic assignment $(\{b,c\})$.
Thus, no lottery assignment is Pareto-efficient for any cardinal utilities consistent with $\succ_1$.
\end{proof}

\subsection{Nonexistence of sd-proportional assignment}
\begin{proposition}\label{prop:no-prop}
A lottery assignment that satisfies sd-proportionality may not exist.
\end{proposition}
\begin{proof}
Consider the following instance with $N=\{1,2\}$ and $E=\{a,b,c\}$:
\begin{itemize}
    \item $\cF_1=\{\emptyset,\{a\},\{b\},\{c\},\{a,b\}\}$,
    \item $\cF_2=\{\emptyset,\{a\},\{b\},\{c\},\{b,c\}\}$,
    \item $b\succ_1 c\succ_1 a$, and
    \item $b\succ_2 a\succ_2 c$.
\end{itemize}
Suppose to the contrary that there exists an sd-proportional lottery assignment $p$.
Since $(\frac{1}{2}\cdot\bm{1},0)$ and $(0,\frac{1}{2}\cdot\bm{1})$ are in $P$, all items must be allocated to agents in each deterministic assignment $\bA\in\supp(p)$.
The deterministic assignments that distribute all the items are only $(\{a,b\},\{c\})$ and $(\{a\},\{b,c\})$.
Since both agents must receive at least half of the item $b$, we have $p_{(\{a,b\},\{c\})}=p_{(\{a\},\{b,c\})}=1/2$.
However, this is not sd-proportional.
\end{proof}

\end{document}